\documentclass[a4paper,12pt]{article}

\usepackage{amsmath,amssymb,mathrsfs}
\usepackage[dvips]{graphicx}
\usepackage{amsthm,bm,
color}
\makeatletter
\def\figcaption{\def\@captype{figure}\caption}
\makeatother

\makeatletter
 
  \@addtoreset{equation}{section}
 \makeatother

\newtheorem{theorem}{\bf Theorem}[section]
\newtheorem{lemma}{\bf Lemma}[section]
\newtheorem{proposition}{\bf Proposition}[section]
\newtheorem{corollary}[theorem]{\bf Corollary}

\newtheorem{example}{\bf Example}[section]
\newtheorem{remark}{\bf Remark}[section]
\newtheorem{definition}{\bf Definition}[section]

\title{Generator of an abstract quantum walk}
\author{
Etsuo Segawa
\thanks{Graduate School of Information Sciences, Tohoku University,
Aoba, Sendai 980-8579, Japan
	},
Akito Suzuki
\thanks{Division of Mathematics and Physics, 
Faculty of Engineering, Shinshu University, Wakasato, Nagano 380-8553, Japan
	}
}
\begin{document}

\maketitle

\begin{abstract}
We give an explicit formula of the generator of 
an abstract Szegedy  evolution operator 
in terms of the discriminant operator of the evolution.
We also characterize the asymptotic behavior of a quantum walker 
through the spectral property of the discriminant operator
by using the discrete analog of the RAGE theorem.
\end{abstract}
\section{Introduction}
Quantum walks (QWs) are one of the interesting topics which have overlaps to various kinds of study fields
(see \cite{Am03, Sh, VA15} and their references). 
While there are several opinions of the priority of QW, 
primitive forms of discrete-time QWs can been seen, for example, 
Feynman and Hibbs~\cite{FH}, Aharonov et al~\cite{Ah}, and Watrous~\cite{Wa}. 
Gudder~\cite{Gud}, Meyer~\cite{Mey}, and Ambainis et al~\cite{Am01} 
introduced the current notion of discrete-time QWs, 
independently. 
The Szegedy walk,
whose original form was introduced in \cite{Sz},
is one of well-investigated discrete-time QWs on graphs.
This includes Grover walk \cite{Gr, Wa} 
and has been intensively studied from various perspectives 
(see, for example, \cite{HKSS13, HKSS14, MNRS,Se}).
Recently, Higuchi et al \cite{HKSS14} introduced
an extended version of the Szegedy walk, the twisted Szegedy walk, 
and proved
a spectral mapping theorem for the new walk on a finite graph.
Using the theorem,
they studied
the spectral and asymptotic properties of the Grover walks 
on crystal lattices.
In our previous paper \cite{SS15},
we studied
an abstract evolution of the form
\begin{equation}
\label{eq_1.1} 
U = S(2d_A^* d_A -1).
\end{equation}
Here $d_A$ is a coisometry from a Hilbert space $\mathcal{H}$
to another Hilbert space $\mathcal{K}$ 
and $S$ is a unitary involution on $\mathcal{H}$:
\[ d_A d_A^* = I_\mathcal{K}, 
	\quad S = S^* = S^{-1}, \]
where $I_\mathcal{K}$ is the identity operator on $\mathcal{K}$.
Let $T = d_A S d_A^*$, called the discriminant of $U$,  
and $\varphi(x) = (x + x^{-1})/2$.
Then the following spectral mapping theorem was proved:
\begin{align*} 
& \sigma(U) 
	= \varphi^{-1}(\sigma(T)) 
		\cup \{+1\}^{M_+} \cup \{-1\}^{M_-}, \\
& \sigma_{\rm p}(U) 
	= \varphi^{-1}(\sigma_{\rm p}(T)) 
		\cup \{+1\}^{M_+} \cup \{-1\}^{M_-},
\end{align*}
where $M_\pm = {\rm dim} \mathcal{D}_\pm^\perp$
and $\mathcal{D}_\pm^\perp = \ker(d_A) \cap \ker(S\pm 1)$.

Let $G = (V, D)$ be a symmetric directed graph,
{\it i.e.}, an arc $e\in D$ if and only if the inverse arc $\bar{e}\in D$.
The evolution $U^{(w,\theta)}$ of the twisted Szegedy walk
on $G$ is of the form
$U^{(w,\theta)} 
= S^{(\theta)} (2 d_A^{(w)*} d_A^{(w)} -1)$,
where 
$d_A^{(w)}:\ell^2(D) \to \ell^2(V)$ 
is a boundary operator 
defined from a weight function $w: D \to \mathbb{C}$
and $S^{(\theta)}$ on $\ell^2(D)$ a (twisted) shift operator
defined from a 1-form $\theta:D \to \mathbb{R}$.
Because 
$U$ becomes $U^{(w,\theta)}$ 
with $d_A = d_A^{(w)}$ and $S = S^{(\theta)}$,
the twisted Szegedy walk on any symmetric directed graph 
is an example of  the abstract Szegedy walks.
In particular, the result of \cite{HKSS14} was extended 
to infinite graphs other than crystal lattices.
The evolution of the Grover walk on $G$ is given by  $U^{(w,\theta)}$
with $w(e) = 1/\sqrt{{\rm deg}(o(e))}$ and $\theta(e) = 0$ ($e \in D$).
In this case, the discriminant $T$ is unitarily equivalent to the transition operator $P_G$
of the symmetric random walk on $G$.
This allows us to determine the the spectrum
of $U^{(w,\theta)}$ from the spectrum of $P_G$
and the subspaces $\mathcal{D}_\pm^\perp$. 

In this paper, we continue the study of the abstract evolution 
$U$ defined by \eqref{eq_1.1}. 
In the case of a continuous-time QW, 
the time evolution is defined as $U(t) = e^{it H}$,
where $H$ is the (negative) Hamiltonian 
(see \cite{CFG} and \cite{Am03} for details).
By the Wiener theorem \cite{Wi} and the RAGE theorem \cite{Ru, AG, En} (see also \cite{RS3}),
the asymptotic behavior of a quantum walker 
is deduced from the spectral properties of $H$.
Motivated by the continuous case,
we give an explicit formula of
the generator $H$ such that 
$H$ is self-adjoint and $U^n = e^{in H}$.
For the evolution $U$ defined by \eqref{eq_1.1}, 
we prove the following.
\begin{itemize}
\item[(1)] The operators
\begin{align*}
& d_+ = \frac{1}{\sqrt{2(1-T^2)} }(d_A -e^{-i\vartheta(T)} d_A S), \\
& d_- = \frac{1}{\sqrt{2(1-T^2)} }(e^{-i\vartheta(T)} d_A - d_A S)
\end{align*} 
can be extended to unitary operators
from ${\rm Ran}(d_\pm^*d_\pm)$ to $\ker(T^2-1)^\perp$.
\item[(2)] Let $\vartheta:[-1,1] \to [0,\pi]$ be the function defined by
$\vartheta(\lambda) = \arccos \lambda$. Then,
the generator $H$ of $U$ is expressed as
\begin{equation*} 
H = \vartheta(d_+^* T d_+) \oplus
	(2\pi - \vartheta(d_-^* T d_-)) \oplus 0 \oplus \pi
\end{equation*} 
on $\mathcal{H} 
	={\rm Ran}(d_+^*) \oplus {\rm Ran}(d_-^*) \oplus \ker(U-1)
	\oplus \ker(U+1)$.
Moreover, 
\[ \ker(U \mp 1) 
	= d_A^* \ker(T \mp 1) \oplus \mathcal{D}_\pm^\perp. \]
\end{itemize}

Let $\mathcal{H}_{\rm p}(A)$ denote the direct sum 
of all eigenspaces of a self-adjoint operator $A$
and $\mathcal{H}_{\sharp}(A)$ ($\sharp = {\rm c, ac, sc}$)
the subspaces of continuity, 
absolute continuity, and singular continuity, respectively.
As a direct consequence of (1) and (2),
the spectral property of the generator $H$ (or the evolution $U$)
is determined 
by the discriminant of $T$ and the subspaces $\mathcal{D}_\pm^\perp$:
\begin{itemize}
\item[(3)] Let $\mathcal{H}_{\rm p}^T 
	:= \mathcal{H}_{\rm p}(T )\cap \ker(T^2-1)^\perp$.
Then,
\begin{align*}
& \mathcal{H}_{\rm p}(H) 
	= d_+^* \mathcal{H}_{\rm p}^T \oplus d_-^* \mathcal{H}_{\rm p}^T 
		\oplus \ker (U^2-1), \\
& \mathcal{H}_{\sharp}(H) 
	=  d_+^* \mathcal{H}_{\rm \sharp}(T)  
		\oplus d_-^* \mathcal{H}_{\rm \sharp}(T).
\end{align*}
\end{itemize}

In what follows,
we consider the long-time asymptotic behavior.
We begin with a general setting.
The setting allows us to introduce the notion of unitary equivalence among QWs
and unify several concrete examples of QWs
such as the Gudder-type QW and the Ambainis-type QW defined in \cite{HKSS13}.
Recently, Ohno proved that
any space-homogeneous QWs on the line \cite{Gud, Mey, Am01} are unitarily equivalent
to abstract Szegedy walks. 
Given a unitary operator $U$ on a Hilbert space $\mathcal{H}$
and a direct sum decomposition $\mathcal{H} = \bigoplus_{x \in V} \mathcal{H}_x$, 
we can naturally introduce 
a directed graph $G_U$ with vertices $V$ and a probability distribution on $V$:
\[ \nu_n(x) = \|P_x U^n \Psi_0\|^2 \quad (x \in V), \]
where $P_x$ is the orthogonal projection onto $\mathcal{H}_x$
and $\Psi_0 \in \mathcal{H}$ is a normalized vector.
We interpret $\nu_n(x)$ as the finding probability of a quantum walker on $G_U$
and $\Psi_0$ as the initial state of the quantum walker. 
In this sense, we say that $U$ is an evolution of QW and 
write $(U, \{\mathcal{H}_x \}_{x \in V}) \in \mathcal{F}_{\rm QW}$.
In the case of the twisted Szegedy evolution $U^{(w,\theta)}$ on $G = (V,D)$,
there is a natural decomposition $\ell^2(D) = \oplus_{x \in V} \mathcal{H}_x$ 
such that $(U^{(w,\theta)}, \{\mathcal{H}_x\}_{x \in V}) \in \mathcal{F}_{\rm QW}$.
Assuming ${\rm dim}\mathcal{H}_x < \infty$ ($x \in V$),
we obtain the discrete analog of the RAGE theorem (see \cite{RS3}):
\begin{itemize}
\item[(4)] $\Psi_0 \in \mathcal{H}_{\rm c}(H)$
if and only if 
$\lim_{N \to \infty} \sum_{n=0}^{N-1} \nu_n^{\Psi_0}(R)/N = 0$
for all finite set $R \subset V$.
\item[(5)] $\Psi_0 \in \mathcal{H}_{\rm p}(H)$
if and only if
$\lim_{m \to \infty} \sup_{n} \nu_n^{\Psi_0}(R_m^{\rm c}) = 0$
for any increasing sequence $\{R_m\}_m$ of finite sets 
such that $\bigcup_m R_m = V$.
\end{itemize}

\noindent
In \cite[Definition 6]{HKSS14}, the authors say that
localization occurs if
\[ \limsup_{n \to \infty} \nu_n^{\Psi_0}(x) > 0
	\quad \mbox{with some $x \in V$.} \]
Let $P_{\rm p}(H)$ be the orthogonal projection onto $\mathcal{H}_{\rm p}(H)$.
Assuming that $\Psi_0 \in \mathcal{H}_{\rm sc}(H)^\perp$,
we observe form (5) the following assertion:
\begin{itemize}
\item[(6)] Localization occurs if and only if
$\Psi_0$ overlaps with $\mathcal{H}_{\rm p}(H)$, {\it i.e.},
$P_{\rm p}(H)\Psi_0$ $\not=0$.
\end{itemize}
For the abstract Szegedy walk,
we know the following from (3) and (6). 
\begin{itemize}
\item[(7)] Localization occurs for some initial state $\Psi_0$
	if and only if $\sigma_{\rm p}(T) \not=\emptyset$ 
		or $\mathcal{D}^\perp \not=\emptyset$.
\item[(8)] If $T$ has a complete set of eigenstates,
	localization occurs for any initial state $\Psi_0$.
\end{itemize}

The remainder of this paper is organized as follows.
Section \ref{sec.1} is devoted to 
the study of the abstract QW.
In Section 2.1, we give the axiom of an  abstract QW
and some concrete examples.
In Section 2.2, we discuss the relation between the generator of an abstract evolution 
and the long-time asymptotic behavior of a quantum walker. 
In particular, we prove (4), (5) and (6).
Section \ref{subsectionevolution} contains 
a brier review of the abstract Szegedy walk.
We summarize the results from \cite{SS15} without proofs.
In Section \ref{sec1.3},
we state the main results of this paper
and prove (7) and (8).
Section \ref{sectiongenerator} is devoted to the derivation of the generator of  
the abstract evolution.
In Subsection \ref{sec.3.1}, 
we present the rigorous definitions of the operators $d_\pm$
and prove (1).
In Subsection  \ref{sec.3.2},
we prove (2) and (3).
In the appendix, we present the proofs of 
the discrete analog of the RAGE theorem
and a relation between the initial state and localization.

\section{
Abstract quantum walks
}
\label{sec.1}
In this section, we first propose QW defied by a unitary operator $U$,
where $U$ is not assumed to be of the form  $U=S(2d_A^*d_A -1)$
but is assumed to act on a Hilbert space
written as a direct sum of Hilbert spaces $\{ \mathcal{H}_v\}_{v \in V}$.
Then, as shown in the following subsection, 
$U$ naturally defines a directed graph $G_U = (V, D)$ and
the probability  of finding a quantum walker thereon.
In addition, we see that 
the dynamics of a quantum walker is governed
by the generator of the evolution $U$. 
\subsection{Axiom of abstract quantum walks}
\label{sec1.1}
Let $V$ be a countable set,
$\{ \mathcal{H}_v \}_{v \in V}$ a family of separable Hilbert spaces
(possibly ${\rm dim}\mathcal{H}_v=\infty$)
and $U$ a unitary on $\mathcal{H} = \bigoplus_{v \in V} \mathcal{H}_v$.
We say that $(U, \{\mathcal{H}_v\}_{v \in V})$
is an evolution of QW
and write $(U, \{\mathcal{H}_v\}_{v \in V}) \in \mathscr{F}_{\rm QW}$.
If there is no danger of confusion,
we simply say that $U$ is an evolution of QW
and write $U \in \mathscr{F}_{\rm QW}$.
We use $P_v$ to denote the projection from $\mathcal{H}$ onto $\mathcal{H}_v$
and define operators $U_{uv}:\mathcal{H}_v \to \mathcal{H}_u$ ($u,v \in V$)
by
\[ U_{uv} = P_u U P_v. \]
First, we introduce a graph associated with $U \in \mathscr{F}_{\rm QW}$.
We use $o(e)$ and $t(e)$ to denote the origin and terminal, respectively, 
of a directed edge $e$ of a graph.
\begin{definition}
{\rm
The graph $G_U = (V_U, D_U)$ associated with an evolution
$(U, \{\mathcal{H}_v\}_{v \in V}) \in \mathscr{F}_{\rm QW}$
is a directed graph defined as follows:
\begin{itemize}
\item[(1)] The set $V_U$ of  vertices of $G_U$ is given by $V_U = V$.
\item[(2)]  If $U_{uv}\neq 0$, there exists an arc $e\in D_U$ from 
$v$ to $u$.
\end{itemize}
}
\end{definition}
Hereafter, we simply write $G_U = (V, D)$
when no confusion can arise.
It is possible that 
depending on the choice of the separation $\{\mathcal{H}_v\}_{v\in V}$, 
there is no inverse {arc} of an {arc} $e \in D$,
because it is not necessary that $U_{vu} \not=0$ even if $U_{uv} \not=0$. 
\begin{example}
{\rm
Let us consider the Hilbert space $\mathcal{H} = \mathbb{C}^3$.
Let $\{\delta_1, \delta_2, \delta_3\}$ be the standard basis of $\mathcal{H}$
and  
\begin{equation*}
U = \begin{pmatrix}
1/\sqrt{2} &  1/\sqrt{2} & 0 \\
0 & 0 & 1 \\
- 1/\sqrt{2} &  1/\sqrt{2} & 0
\end{pmatrix}
\end{equation*}
a unitary matrix on $\mathcal{H}$.
\begin{itemize}
\item[(i)] Let $V =\{ a, b\}$.
We consider the separation $\{\mathcal{H}_a, \mathcal{H}_b\}$ of $\mathcal{H}$,
where $\mathcal{H}_a ={\rm Span}\{ \delta_1 \}$ and
$\mathcal{H}_b ={\rm Span}\{\delta_2, \delta_3\}$. 
By this separation, $U$ is decomposed as
\[
U = \left(\begin{array}{c|cc}
1/\sqrt{2} &  1/\sqrt{2} & 0 \\ \hline 
0 & 0 & 1 \\
- 1/\sqrt{2} &  1/\sqrt{2} & 0
\end{array}
\right).
\]
Hence, $G_U$ has an arc from $a$ to $b$
and its inverse arc. 
$G_U$ has loops at $a$ and $b$.
\item[(ii)]  Let $V = \{a, b,c\}$ and consider 
the separation  $\{\mathcal{H}_v\}_{v \in V}$, 
where
$\mathcal{H}_a = {\rm Span}\{\delta_1\}$,
$\mathcal{H}_b = {\rm Span}\{\delta_2\}$,
and $\mathcal{H}_c = {\rm Span}\{\delta_3\}$.
$U$ is decomposed as
\[
U = \left(\begin{array}{c|c|c}
1/\sqrt{2} &  1/\sqrt{2} & 0 \\ \hline 
0 & 0 & 1 \\ \hline
- 1/\sqrt{2} &  1/\sqrt{2} & 0
\end{array}
\right).
\]
We observe that $U_{ba} = U_{ca}= 0$, 
whereas $U_{ab}\not=0$ and $U_{ac}\not=0$.
Hence, $G_U$ has no inverse arcs of an arc from $b$ to $a$ 
and an arc from $c$ to $a$.
$G_U$ has an arc from $b$ to $c$, its inverse arc, and a loop only at $a$. 
\end{itemize}
}
\end{example}
In the following, we introduce an abstract QW on $G_U$.
\vspace{3mm}

\noindent
{\bf Axiom.} QW with an evolution $(U, \{\mathcal{H}_v\}_{v \in V})$
$\in \mathscr{F}_{\rm QW}$
is defined as follows:
{
\begin{itemize}
\item[(1)]
 	The state of a quantum walker
	at time $n \in \mathbb{N}$ with the initial state $\Psi_0 \in \mathcal{H}$ 
	($\|\Psi_0\|=1$) is given by $\Psi_n = U^n \Psi_0$. 
\item[(2)] The probability $\nu_n(x)$ of  finding the quantum walker 
	at vertex $x \in V$ at time $n \in \mathbb{N}$ is given by
	$\nu_n(x) = \| P_x  \Psi_n \|^2$.
\end{itemize}
}
\begin{example}
\label{ex_homogeneous}
{\rm
The evolution of a typical QW on $\mathbb{Z}$ is of the form
\[ U = \sum_{x \in \mathbb{Z}} 
	\left( |x+1\rangle \langle x| \otimes Q
		+ |x-1\rangle \langle x| \otimes P \right), \]
which converges in the strong operator topology.
Here, $P, Q \in M_2(\mathbb{C})$
and the Hilbert space of states is given by 
$\mathcal{H} = \ell^2(\mathbb{Z}) \otimes \mathbb{C}^2$.
Noting that $\mathcal{H} = \oplus_{x \in \mathbb{Z}} \mathcal{H}_x$
with 
$\mathcal{H}_x = {\rm Ran}( |x \rangle \langle x | \otimes I_{\mathbb{C}^2} )
	\simeq \mathbb{C}^2$,
we see that 
\begin{equation}
\label{Uyx}
U_{yx} 
	= \begin{cases} 
	|y \rangle \langle x| \otimes P, & y= x-1,\\
	|y \rangle \langle x| \otimes Q, & y= x+1,\\
	0, & \mbox{otherwise}.
	\end{cases}
\end{equation}
We observe from Proposition \ref{lemmacriteria} below that
$U$ is unitary if and only if $P$ and $Q$ satisfy
\begin{equation}
\label{eqPQ}
PP^* + QQ^* = P^*P + Q^*Q = 1,
	\quad PQ^* = Q^*P = 0. 
\end{equation}
For example, if $P = \begin{pmatrix} a & b \\ 0 & 0 \end{pmatrix}$,
$Q = \begin{pmatrix} 0 & 0 \\ c & d \end{pmatrix}$
and $P+Q$ is unitary,
$P$ and $Q$ satisfy \eqref{eqPQ}.
Hence,
$(U, \{\mathcal{H}_x\}_{x \in \mathbb{Z}}) \in \mathscr{F}_{\rm QW}$
and the graph $G_U$ associated with $U$ is 
the symmetric directed graph of $\mathbb{Z}$.
Because $P_x = |x \rangle \langle x | \otimes I_{\mathbb{C}^2}$,
we know that the probability of finding a quantum walker at vertex $x \in \mathbb{Z}$ 
at time $n \in \mathbb{N}$
with an initial state $\Psi_0 \in \mathcal{H}$ is 
$\nu_n(x) = \|\Psi_n(x)\|_{\mathbb{C}^2}^2$.
For a deeper discussion of this QW, we refer the reader to \cite{Am01, Am03}.
}
\end{example}

\begin{proposition}
\label{lemmacriteria}
{\rm 
Let $W$ be a bounded operator on $\mathcal{H} = \oplus_{v \in V}\mathcal{H}_v$
and $W_{uv} = P_u W P_v$ ($u,v \in V$).
The following are equivalent:
\begin{itemize}
\item[(i)] $W$ is unitary.
\item[(ii)] $\sum_{x \in V} W_{ux} (W^*)_{xv} 
	= \sum_{x \in V} (W^*)_{ux} W_{xv} = \delta_{uv} P_v$ 
		for all $u,v \in V$.
\end{itemize}
}
\end{proposition}
\begin{proof}
The operator equality
$I = \sum_{v \in V} P_v$
and the equalities
\begin{align*}
(WW^*)_{uv} 
= \sum_{v \in V}  W_{ux} (W^*)_{xv}
\quad \mbox{and} \quad
(W^*W)_{uv} 
= \sum_{v \in V}  (W^*)_{ux} W_{xv}
\end{align*}
all hold  in the strong convergence sense.
Hence, (ii) is equivalent to $WW^* = W^*W = I_{\mathcal{H}}$,
which proves the proposition.
\end{proof}

\begin{definition}
\label{defUV}
{\rm
$(U_1, \{\mathcal{H}^{(1)}_{v_1}\}_{v \in V_1}) \in \mathscr{F}_{\rm QW}$ 
and $(U_2, \mathcal{H}^{(2)}_{v_2}\}_{v \in V_2}) \in \mathscr{F}_{\rm QW}$
are unitarily equivalent,
written 
$(U_1, \{\mathcal{H}^{(1)}_{v_1}\}_{v \in V_1}) 
	\simeq (U_2, \mathcal{H}^{(2)}_{v_2}\}_{v \in V_2})$,
if there exist a unitary 
$\mathscr{U}:\bigoplus_{v_1 \in V_1} \mathcal{H}_{v_1}
\to \bigoplus_{v_2 \in V_2} \mathcal{H}_{v_2}$ 
and a bijection $\phi:V_1 \to V_2$
such that $\mathscr{U} \mathcal{H}^{(1)}_{v_1} = \mathcal{H}^{(2)}_{\phi(v_1)}$
and $\mathscr{U} U_1 \mathscr{U}^{-1} = U_2$.
}
\end{definition}
Let $(U_1, \{\mathcal{H}^{(1)}_{v_1}\}_{v_1 \in V_1}) \in \mathscr{F}_{\rm QW}$ 
and $(U_2, \{\mathcal{H}^{(2)}_{v_2}\}_{v_2 \in V_2}) \in \mathscr{F}_{\rm QW}$
be unitarily equivalent.
The state 
$U_1^n\Psi_0^{(1)} \in \mathcal{H}_1:=\bigoplus_{v \in V_1} \mathcal{H}^{(1)}_v$ 
of a quantum walker at time $n \in \mathbb{N}$
 is identified with 
$U_2^n\Psi_0^{(2)} = \mathscr{U}  (U_1^n\Psi_0^{(1)} )
	\in \mathcal{H}_2:=\bigoplus_{v \in V_2} \mathcal{H}^{(2)}_v$,
where $\Psi_0^{(2)} = \mathscr{U} \Psi_0^{(1)}$. 
Since 
$\mathscr{U} \mathcal{H}^{(1)}_{v_1} = \mathcal{H}^{(2)}_{\phi(v_1)}$,
we have $P_{\phi(v_1)} = \mathscr{U} P_{v_1} \mathscr{U}^{-1}$.
Hence,
the probability $\nu^{(1)}_n(x_1) := \|P_{x_1} \Psi_n^{(1)}\|^2$ 
of finding a quantum walker at vertex $x_1 \in V_1$ and at time $n \in \mathbb{N}$ 
is equal to $\nu_n^{(2)}(\phi(x_1)) :=\|P_{\phi(x_1)} \Psi_n^{(2)} \|^2$.
We also know that the bijection $\phi:V_1 \to V_2$ is an isomorphism
between the associated graphs $G_{U_1}$ and $G_{U_2}$.
\begin{proposition}
\label{prop_iso}
{\rm
Let $W_1$ and $W_2$ be unitary operators on 
$\mathcal{H} = \oplus_{v \in V} \mathcal{H}_v$
and set $U =W_1W_2$ and $\tilde{U} = W_2W_1$.
Then, 
\[ (U, \{\mathcal{H}_v\}) \simeq (\tilde U, \{W_2 \mathcal{H}_v\})
	\simeq (U, \{W_1^* \mathcal{H}_v\}). \]
}
\end{proposition}
\begin{proof}
Let $\mathscr{U}=W_2$ and $\phi$ be an identity map on $V$.
Then,  $\mathscr{U}\mathcal{H}_v = W_2 \mathcal{H}_v$
and $\mathscr{U} U \mathscr{U}^{-1} = W_2(W_1W_2)W_2^{-1}=\tilde U$.
Hence, $(U, \{\mathcal{H}_v\}) \simeq (\tilde U, \{W_2 \mathcal{H}_v\})$.
Similarly, we know that
$(\tilde U, \{W_2 \mathcal{H}_v\})
	\simeq (U, \{W_1^* \mathcal{H}_v\})$
if we take $\mathscr{U} = W_1^*$.
\end{proof}
\begin{example}[Gudder and Ambainis type QWs]
{\rm
Here, we follow the notation of \cite{HKSS13}.
Let $S_\pi$ be a shift operator and $C=\oplus_{j \in V(\mathcal{G})} H_j$ a coin flip operator,
where $\pi$ is a partition on the line digraph of  a graph $\mathcal{G}$ 
and $\{H_j\}$ is a sequence of unitary operators on $\mathcal{H}_j$. 
Note that $C\mathcal{H}_j = \mathcal{H}_j$.
We observe, from Proposition \ref{prop_iso}, that
the Gudder type evolution $U^{(G)} = C S_\pi$
and the Ambainis type evolution $U^{(A)} = S_\pi C$
are unitarily equivalent and
\[ (U^{(G)}, \{\mathcal{H}_j\}) \simeq (U^{(A)}, \{\mathcal{H}_j\}). \]
}
\end{example}

\subsection{Generators}
It is well known that for a unitary operator $U$,
there exists a unique self-adjoint operator $H$
such that  
\begin{equation}
\label{eqUH} 
 E_H([0,2\pi)) = I
	\quad \mbox{and} \quad U= e^{i H},
\end{equation} 
where $E_H$ is the spectral measure of $H$.
The state of a quantum walker at time $n \in \mathbb{N}$
is represented as $\Psi_n = e^{i n H} \Psi_0$ ($n \in \mathbb{N}$).
In this sense, we define the generator of a unitary operator 
as follows:
\begin{definition}
{\rm
A self-adjoint operator
$H$ is the generator of a unitary operator $U$,
if \eqref{eqUH} holds.
}
\end{definition}
Let $H$ be the generator of an evolution 
$(U, \{\mathcal{H}_v\}_{ v \in V}) \in \mathscr{F}_{\rm QW}$.
Then, the probability $\nu_n(x)$ of  finding a quantum walker
at vertex $x \in V$ at time $n \in \mathbb{N}$ is given by
	\[ \nu_n(x) = \| P_x  e^{i n H} \Psi_0 \|^2. \]
We denote by $\nu_n(R)$ the probability of 
finding a quantum walker in $R \subset V$:
\[ \nu_n(R) = \sum_{x \in R} \nu_n(x). \]
We denote  $\nu_n(x)$ (resp., $\nu_n(R)$) 
by $\nu_n^{\Psi_0}(x)$ (resp., $\nu_n^{\Psi_0}(R)$) 
to emphasize the dependence on the initial state.
The time average $\bar\nu_N^{\Psi_0}$ of $\nu_n$ 
and its infinite time limit $\bar\nu_\infty^{\Psi_0}$, if it exists,
are given by
\[ \bar\nu_N^{\Psi_0}(R)
	= \frac{1}{N}\sum_{n=0}^{N-1} \nu_n^{\Psi_0}(R)
	\quad \mbox{and}
	\quad \bar\nu_\infty^{\Psi_0}(R) = \lim_{N \to \infty} \bar\nu_N^{\Psi_0}(R). \]
\begin{proposition}
\label{proplocal}
{\rm
Let $H$ be the generator of an evolution 
$(U, \{\mathcal{H}_v\}_{v \in V}) \in \mathscr{F}_{\rm QW}$,
and assume that ${\rm dim} \mathcal{H}_v < \infty$ ($v \in V$).
\begin{itemize}
\item[(i)] $\Psi_0 \in \mathcal{H}_{\rm c}(H)$
if and only if 
$\bar\nu_\infty^{\Psi_0}(R) = 0$
for all finite sets $R \subset V$.
\item[(ii)] $\Psi_0 \in \mathcal{H}_{\rm p}(H)$
if and only if
$\lim_{m \to \infty} \sup_{n} \nu_n^{\Psi_0}(R_m^{\rm c}) = 0$
for any increasing sequence $\{R_m\}_m$ of finite sets 
such that $\bigcup_m R_m = V$.
\end{itemize}
}
\end{proposition}%
\noindent
This proposition is the discrete analog of the RAGE theorem.
The proof is standard, 
but we include it  in the appendix for completeness.

In \cite[Definition 6]{HKSS14}, the authors say that
localization occurs if 
\begin{equation} 
\label{eqloc}
\limsup_{n \to \infty} \nu_n^{\Psi_0}(x) > 0
	\quad \mbox{with some $x \in V$.}
\end{equation}
As will be proved in the appendix,
\eqref{eqloc} holds 
if $\lim_{m \to \infty} \sup_{n} \nu_n^{\Psi_0}(R_m^{\rm c})$
$ = 0$
for some increasing sequence $\{R_m\}$ such that  $\bigcup_m R_m = V$.
Hence, localization occurs if $\Psi_0 \in \mathcal{H}_{\rm p}(H)$.

Let $P_{\sharp}(H)$ be the orthogonal projection onto $\mathcal{H}_{\sharp}(H)$
for $\sharp = {\rm p, ac}$.
\begin{proposition}
\label{prop_loc}
{\rm
Let $H$ and
$(U, \{\mathcal{H}_v\}_{v \in V}) \in \mathscr{F}_{\rm QW}$
be as in Proposition \ref{proplocal}.
Suppose that $\Psi_0 \in \mathcal{H}_{\rm sc}(H)^\perp$.
Then, the following are equivalent:
\begin{itemize}
\item[(a)] Localization occurs.
\item[(b)] $\Psi_0$ overlaps with $\mathcal{H}_{\rm p}(H)$, 
{\it i.e.}, $P_{\rm p}(H)\Psi_0 \not=0$.
\end{itemize}
}
\end{proposition}
\begin{proof}
By assumption, we can write $\Psi_0 = \Psi_{\rm p} + \Psi_{\rm ac}$
with $\Psi_\sharp =P_{\sharp}(H)\Psi_0$ ($\sharp = {\rm p, ac}$).
Because, by assumption, $P_x$ is compact,
$P_x U^n \Psi_{\rm ac}$ converges strongly to zero as $n \to \infty$.
Hence,
\begin{align*} 
|\nu_n^{\Psi_0}(x)^{1/2} - \nu_n^{\Psi_{\rm p}}(x)^{1/2}|
\leq \|P_x U^n \Psi_{\rm ac}\| \to 0 \quad (n \to \infty).
\end{align*}
Assuming (a), we get $\epsilon_0 := \limsup_{n \to \infty}\nu_n^{\Psi_0}(x) > 0$ 
with some $x \in V$.
Because
$\nu_n^{\Psi_{\rm p}}(x) \geq \nu_n^{\Psi_0}(x) - \epsilon_0/2$ for sufficiently large $n$,
\[ \|\Psi_{\rm p}\|^2
	\geq \nu_n^{\Psi_0}(x) - \epsilon_0/2. \]
Taking the limit superior on both sides, we have (b).
Conversely, we assume (b).
Then, $\Psi_{\rm p} \not=0$.
By the above argument, 
$\epsilon_1 := \limsup_{n \to \infty}\nu_n^{\Psi_{\rm p}}(x) > 0$ 
with some $x \in V$.
Because $\nu_n^{\Psi_0}(x) \geq \nu_n^{\Psi_{\rm p}}(x) - \epsilon_1/2$ 
for sufficiently large $n$,
\[ \limsup_{n \to \infty}\nu_n^{\Psi_0}(x) 
	\geq \epsilon_1/2 > 0. \]
This proves (a).
\end{proof}

\section{Abstract Szegedy walk}
\label{subsectionevolution}
In this section, we treat a specific class of abstract QWs, 
an extension of the Szegedy walks. 
Let us recall some notations and facts from \cite{SS15}.
Let $\mathcal{H}$ and $\mathcal{K}$ be complex Hilbert spaces. 
We assume that there exists a coisometry operator $d_A:\mathcal{H} \to \mathcal{K}$, {\it i.e.}, $d_A$ is bounded and satisfies
\begin{equation}
\label{ddstar}
d_A d_A^* = I_{\mathcal{K}}, 
\end{equation}
where $I_{\mathcal{K}}$ is the identity operator on $\mathcal{K}$.
By \eqref{ddstar}, $d_A$ is a partial isometry and surjection, its adjoint $d_A^*:\mathcal{K} \to \mathcal{H}$ is an isometry, and $\Pi_\mathcal {A}:= d_A^*d_A$
is the projection onto  $\mathcal{A}:={\rm Ran} (d_A^* d_A) = d_A^* \mathcal{K}$.
We call the self-adjoint operator $C := 2 d_A^* d_A - 1$ on $\mathcal{H}$ 
a {\it coin operator}, 
because we observe that
$C$ is a unitary involution and decomposed into 
\[ C = I_\mathcal{A}  \oplus (-I_{\mathcal{A}^\perp }) \quad \mbox{on 
$\mathcal{H} =\mathcal{A} \oplus \mathcal{A}^\perp$.} \]
This also proves that $\mathcal{A} = \ker(C-1)$ and $\mathcal{A}^\perp = \ker(C+1)$.

Let $S$ be a unitary involution on $\mathcal{H}$. 
We decompose $S$ into $S=I_\mathcal{S}\oplus (-I_{\mathcal{S}^\bot})$ on $\mathcal{H}=\mathcal{S}\oplus \mathcal{S}^\bot$, where
$\mathcal{S}=\mathrm{ker}(S-1)$ and $\mathcal{S}^\perp=\mathrm{ker}(S+1)$.
Then $d_B := d_A S$ is also a coisometry.
Throughout this subsection, we fix $d_A$ and $S$,
and call them a {\it boundary operator} and a {\it shift operator}, respectively.
In analogy with the twisted Szegedy walk (see Example \ref{ExampleTSW} below),
we define an abstract evolution $U$ 
and its {\it discriminant} $T$ as follows:
\begin{definition}
\label{def04062144}
{\rm 
Let $d_A$, $d_B$, $C$, and $S$ be as above.
\begin{itemize}
\label{abstSz}
\item[(1)] The evolution associated with 
	the boundary operator $d_A$ and the shift operator $S$ is defined by $U = S C$.  
\item[(2)] The discriminant of $U$ is defined by $T = d_A d_B^*$. 
\end{itemize}
}
\end{definition}
We note that
$S$, $C$, and $U$ are unitary on $\mathcal{H}$.
By definition, the discriminant $T$ 
is a bounded self-adjoint operator on $\mathcal{K}$
with $\|T\| \leq 1$.
Let 
\begin{equation}
\label{Dpmperp} 
\mathcal{D}^\perp_+ = \mathcal{A}^\perp \cap \mathcal{S}^\perp,
	\quad \mathcal{D}^\perp_- = \mathcal{A}^\perp \cap \mathcal{S}. 
\end{equation}
\begin{theorem}[\cite{SS15}]
\label{thm0409}
{\rm
Let $M_\pm = {\rm dim}\mathcal{D}^\perp_\pm$.
\begin{itemize}
\item[(1)] $\sigma(U) = \{ e^{i\xi} \mid \cos \xi \in \sigma(T), 
	\xi \in [0,2\pi) \} \cup \{+1\}^{M_+} \cup \{-1\}^{M_-}$;
\item[(2)]  $\sigma_{\rm p}(U) = \{ e^{i\xi} \mid \cos \xi \in \sigma_{\rm p}(T), 
	\xi \in [0,2\pi) \} \cup \{+1\}^{M_+} \cup \{-1\}^{M_-}$,
\end{itemize}
where we use $\{\pm 1\}^{M_\pm}$ to denote the multiplicity of $\pm 1$
and set $\{\pm 1\}^{M_\pm} = \emptyset$ if $M_\pm = 0$.
}
\end{theorem}

\begin{example}[Twisted Szegedy walk \cite{HKSS14}]
\label{ExampleTSW}
{\rm
Let $G =(V, E)$ be a (possibly infinite) graph
with the sets $V$ of vertices and $E$ of unoriented edges
(possibly including multiple edges and loops). 
We consider that each edge $e \in E$ with end vertices $V(e) = \{ u,v\}$ has  two orientations such that the origin of $e$ is $u$ or $v$,
and we denote the set of such oriented edges by $D$.
For each edge $e \in D$, we use $o(e)$ (resp. $t(e)$) to denote 
the origin (resp. terminal) of $e \in D$.
The inverse edge of $e \in D$ 
is denoted by $\bar{e}$, 
with the result that $o(\bar{e}) = t(e)$ and $t(\bar{e}) = o(e)$. 
Note that $e\in D$ if and only if $\bar{e}\in D$.
Let $\mathcal{H} = \ell^2(D)$ and $\mathcal{K} = \ell^2(V)$.
We define a {\it boundary operator} $d_A^{(w)}:\mathcal{H} \to \mathcal{K}$
as follows.
We call $w:D \to \mathbb{C}\setminus\{0\}$ a {\it weight} if it satisfies $w(e) \not=0$ 
and
\begin{equation}
\label{eqweight}
\sum_{e:o(e) = v} |w(e)|^2 = 1 \quad \mbox{for all $v \in V$}. 
\end{equation}
For a weight $w$ and all $\psi \in \mathcal{H}$, $d_A^{(w)}\psi \in \mathcal{K}$ 
is given by
\begin{align*}
(d_A^{(w)}\psi)(v) = \sum_{e:o(e)=v} \psi(e) \overline{w(e)}, \quad v \in V.
\end{align*}
The adjoint $d_A^{(w)*}:\mathcal{K} \to \mathcal{H}$ of $d_A^{(w)}$
is a {\it coboundary operator} and satisfies
\begin{align*}
(d_A^{(w)*}f)(e) = w(e) f(o(e)), \quad e \in D
\end{align*} 
for all $f \in \mathcal{K}$.
We observe that  $d_A^{(w)}$ is a coisometry, {\it i.e.}, 
$d_A^{(w)} d_A^{(w)*} = I_\mathcal{K}$,
because, from \eqref{eqweight},
\[ (d_A^{(w)}d_A^{(w)*}f)(v) = \sum_{e:o(e)=v} (d_A^{(w)*}f)(e) \overline{w(e)}
	 = \sum_{e:o(e)=v}  |w(e)|^2 f(o(e)) = f(v). \]
The coin operator is defined by $C^{(w)} = 2 d_A^{(w)*} d_A^{(w)} - 1$,
and  the (twisted) shift operator by
$(S^{(\theta)}\psi)(e) = e^{-i\theta(e)}\psi(\bar{e})$ ($e \in D$),
where $\theta:D \to \mathbb{R}$ is
a 1-form and satisfies
$\theta(\bar{e}) = - \theta(e)$ ($e \in D$). 
It is easy to check that $S^{(\theta)}$ is a unitary involution.
The evolution of the twisted Szegedy walk
associated with the weight $w$ and the 1-form $\theta$
is defined by $U^{(w,\theta)} = S^{(\theta)} C^{(w)}$.
The operators $d_A^{(w)}$, $C^{(w)}$, 
and $S^{(\theta)}$
are examples of the abstract coisometry $d_A$, coin operator $C$,
and shift operator $S$, respectively.
The discriminant of $U^{(w,\theta)}$ is defined by 
$ T^{(w,\theta)} = d_A^{(w)} d_B^{(w,\theta)*}$,
where $d_B^{(w,\theta)} = d_A^{(w)} S^{(\theta)}$.
We now show that $U^{(w,\theta)}$ is an evolution of QW.
To this end, we set
\begin{equation}
\label{eq050802:05} 
\mathcal{H}_v 
	= \overline{\rm Span}\left\{  \delta_e 
		\mid e \in D, o(e) = v \right\}, 
\end{equation}
where $\overline{\rm Span} A$ is 
the closure of the linear span of a set $A$
and $\delta_e \in \ell^2(D)$ is given by $\delta_e(e) = 1$ 
and $\delta_e(f) = 0$ ($e \not= f$).
Then, we can decompose $\mathcal{H}$ into 
$\mathcal{H} = \bigoplus_{ v \in V} \mathcal{H}_v$. 
Thus, we know that 
$(U^{(w,\theta)}, \{\mathcal{H}_v\}_{v \in V}) \in \mathscr{F}_{\rm QW}$.
Observe that the orthogonal projection onto $\mathcal{H}_v$ is given by
\[ P_v = \sum_{e \in D: o(e) = v} | e \rangle \langle e|, \]
where $| e \rangle \langle e| = \langle \delta_e, \cdot \rangle \delta_e$ 
is the orthogonal projection onto 
the one dimensional subspace $\{ \alpha \delta_e \mid \alpha \in \mathbb{C} \}$. 
The probability $\nu_n: V\to [0,1]$ of finding a quantum walker at time $n$ is
\begin{align*} 
\nu_n(x) & 
	= \sum_{e \in D: o(e) = x} 
		|\langle \delta_e, \Psi_n \rangle|^2
	 = \sum_{e \in D: o(e) = x} |\Psi_n (e)|^2.
\end{align*}
Let $G_U$ be the associated graph of $U$.
We observe that
\[ U_{uv} = \sum_{e:o(e) = u, t(e)=v} \sum_{f:o(f)=v} 
	w(\bar{e}) \overline{w(f)} (2-\delta_{f\bar{e}})
		 e^{i \theta(\bar{e})} |e\rangle \langle f|  \] 
is non-zero if and only if there exists $e \in D$ such that $o(e) =u$ and $t(e)=v$.
Hence, $G_U$ is identified with a subgraph of $G$.
If $G$ has no multiple edges, $G_U \simeq G$.
}
\end{example}

\begin{remark}
Recently, Ohno proved that
any space-homogeneous QW on $\mathbb{Z}$ such as the model in 
Example \ref{ex_homogeneous} is
equivalent to an abstract QW 
associated with some boundary operator and shift operator.
Even for an inhomogeneous case such as \cite{K,ShK},
we can show that the evolution is associated with some boundary operator 
and shift operator. 
\end{remark}

\begin{remark}
We should remark that 
for a time-dependent abstract Szegedy walk $U_1\to U_2\to \cdots \to U_n$, 
the spectrum of $U_{n}U_{n-1}\cdots U_1$ cannot be described by the discriminant operators $T_nT_{n-1}\cdots T_1$ in general, 
since our analysis proposed here essentially works well when the time evolution is decomposed into two involution operators 
$U=E_2E_1$~\cite{SS15}. 
Applying our abstract QW to the time-dependent QW effectively is an open problem. 
\end{remark}

In what follows, we introduce closed subspaces of $\mathcal{H}$ 
that play an important role in this paper:
\begin{align*}
\mathcal{D} = \overline{\mathcal{A} + \mathcal{B}},\quad
\mathcal{D}_0 = \mathcal{A} \cap \mathcal{B}, \quad
\mathcal{D}_1 = \mathcal{D}_0^\perp \cap \mathcal{D}. 
\end{align*}
Here, we denote by $\mathcal{A}$ and $\mathcal{B}$ the subspaces ${\rm Ran}(d_A^*d_A)$ and ${\rm Ran}(d_B^*d_B)$, respectively.
Clearly,
\begin{align*} 
\mathcal{H} & = \mathcal{D} \oplus  \mathcal{D}^\perp \\
	& = \mathcal{D}_1 \oplus  \mathcal{D}_0 \oplus  \mathcal{D}^\perp.
\end{align*}

We state the basic properties of these subspaces without proof.
For the proof, one can consult  \cite{SS15},
where we used the notations $\mathcal{L}$, $\mathcal{L}_1$, and $\mathcal{L}_0$
with $\mathcal{D} = \overline{\mathcal{L}}$,
$\mathcal{D}_1 = \overline{\mathcal{L}_1}$,
and $\mathcal{D}_0 = \mathcal{L}$. 
\begin{proposition}
\label{prop42910:04}
{\rm
Let $U$ be as above and $T = d_Ad_B^*$ the discriminant of $U$.
$U$ leaves $\mathcal{D}$, $\mathcal{D}_1$, $\mathcal{D}_0$,
and $\mathcal{D}^\perp$ invariant.
Moreover, the following hold:
\begin{itemize}
\item[(i)] $\mathcal{D}_0 =  d_A^* \ker(T^2-1) = d_B^* \ker(T^2-1)$;
\item[(ii)] $\mathcal{D}_1 
		= \overline{d_A^* \ker(T^2-1)^\perp + d_B^*\ker(T^2-1)^\perp}$;
\item[(iii)] $\mathcal{D}^\perp = \ker (d_A) \cap \ker (d_B)$.
\end{itemize}
}
\end{proposition}
By Proposition \ref{prop42910:04},  $U$ is decomposed as
\begin{equation}
\label{Udecomp} 
U = U_{{\mathcal{D}_1}} 
		\oplus U_{\mathcal{D}_0}  \oplus U_{\mathcal{D}^\perp}. 
\end{equation}
Since $\ker(T^2-1) = \ker(T-1) \oplus \ker(T+1)$,
we know that
\begin{align*} 
\mathcal{D}_0 = \mathcal{D}_0^+ \oplus \mathcal{D}_0^- , 
\end{align*}
where $\mathcal{D}_0^\pm = d_A^*\ker(T\mp 1)$.
We also have
\[ \mathcal{D}^\perp
	=\mathcal{D}^\perp_+ \oplus \mathcal{D}^\perp_-, \]
where $\mathcal{D}_\pm^\perp:= \mathcal{D}^\perp \cap \ker (S\mp 1)$.
{
By Proposition \ref{prop42910:04} (iii),
we have \eqref{Dpmperp}.
}
The following  is essentially proved in \cite{SS15}.
\begin{proposition}
\label{prop0409}
{\rm
Let $M_\pm = {\rm dim}\mathcal{D}^\perp_\pm$. 
\begin{itemize}
\item[(1)] $\ker(U \mp 1) = \mathcal{D}_0^\pm \oplus \mathcal{D}_\pm^\perp$
	and $\ker(U^2 - 1)^\perp = \mathcal{D}_1$;
\item[(2)] $U_{\mathcal{D}_0} = I_{\mathcal{D}_0^+} \oplus (- I_{\mathcal{D}_0^-})$
		and $U_{\mathcal{D}^\perp} = I_{\mathcal{D}_+^\perp} 
			\oplus (-I_{\mathcal{D}_-^\perp})$.
\end{itemize}
}
\end{proposition}

\section{Main results}
\label{sec1.3}
Let $U = S(2 d_A^* d_A -1)$ be an evolution
associated with a boundary operator $d_A :\mathcal{H} \to \mathcal{K}$
and a shift operator $S$ on $\mathcal{H}$.
As will be seen in Section \ref{sectiongenerator},
the operators 
\begin{align*}
d_+ = \frac{1}{\sqrt{2(1-T^2)} }(d_A -e^{-i\vartheta(T)} d_B), \quad
	d_- = \frac{1}{\sqrt{2(1-T^2)} }(e^{-i\vartheta(T)} d_A - d_B),
\end{align*} 
can be extended to bounded operators,
where $T$ is the discriminant of $U$
and 
$\vartheta:[-1,1] \to  [0,\pi]$
is given by
$\vartheta(\lambda) = \arccos \lambda$.
{
\begin{theorem}
\label{mainthm01}
{\rm
Let $U$, $d_\pm$ and $T$ be as above.
Then, $\mathcal{H}$ is decomposed as
\begin{equation}\label{ibuki} \mathcal{H} 
	= {\rm Ran}(d_+^* d_+) \oplus {\rm Ran}(d_-^* d_-)
		\oplus \ker(U-1) \oplus \ker(U+1) \end{equation}
and the generator $H$ of  $U$ is given by
\begin{equation}
\label{eqHform}  
H
	= \vartheta(d_+^* T  d_+) 
		\oplus (2\pi -\vartheta(d_-^* T d_-)) \oplus 0 \oplus \pi,
\end{equation}
where 
\[ \ker(U \mp 1)
	= d_A^* \ker (T \mp 1) \oplus \mathcal{D}_\pm^\perp.
	\] 
}
\end{theorem}
}
By this theorem, $U$ is expressed by 
\begin{equation}
\label{eq822} 
U=e^{i\vartheta(d_+^*Td_+)}\oplus e^{-i\vartheta(d_-^*Td_-)} \oplus 1\oplus (-1) 
\end{equation}
under the decomposition of (\ref{ibuki}). 
We consider the iteration of $U$, $\psi_0\stackrel{U}{\to}\psi_1\stackrel{U}{\to}\psi_2\stackrel{U}{\to}\cdots$.
From \eqref{eq822}, we obtain 
the following temporal and spatial discrete analog of the wave equation.
\begin{corollary}
\label{mainthm01.5}
{\rm
Let $\psi_0\in \mathrm{Ran}(d_+^*d_+)$
and $f_n=d_+\psi_n$. Then, 
	\[ \frac{1}{2}\left(f_{n+1}+f_{n-1}\right)=Tf_n. \]
}
\end{corollary}
Moreover, we obtain the following corollary, 
which is important for discriminating the localization of QW 
under the time evolution $U$. 
\begin{corollary}
\label{mainthm02}
{\rm
Let $U$, $d_\pm$, $T$ and $H$ be as in Theorem \ref{mainthm01}.
Then, 
\begin{align*}
& \mathcal{H}_{\rm p}(H) 
	= d_+^* \mathcal{H}_{\rm p}^T \oplus d_-^* \mathcal{H}_{\rm p}^T 
		\oplus \ker (U^2-1), \\
& \mathcal{H}_{\sharp}(H) 
	=  d_+^* \mathcal{H}_{\sharp}(T)  \oplus d_-^* \mathcal{H}_{\sharp}(T),
	\quad \sharp = {\rm c, ac, sc}
\end{align*}
where $\mathcal{H}_{\rm p}^T := \mathcal{H}_{\rm p}(T )\cap \ker(T^2-1)^\perp$.
}
\end{corollary}
Combining Corollary \ref{mainthm02} with Proposition \ref{prop_loc}, 
we have the following criterion for localization.
\begin{corollary}
\label{mainthm03}
{\rm
Let $U =S(2d_A^* d_A -1)$ and $H$ be as in Theorem \ref{mainthm01}.
Assume that there exists a family $\{\mathcal{H}_v\}_{v \in V}$ of Hilbert spaces
such that $(U, \{\mathcal{H}_v\}_{v \in V}) \in \mathscr{F}_{\rm QW}$
and ${\rm dim}\mathcal{H}_v < \infty$ ($v \in V$). 
Then,
\begin{itemize}
\item[(1)] Localization occurs for some initial state $\Psi_0$
	if and only if $\sigma_{\rm p}(T) \not=\emptyset$ 
		or $\mathcal{D}^\perp \not=\emptyset$.
\item[(2)] If $T$ has a complete set of eigenstates,
	localization occurs for any initial state $\Psi_0$.
\end{itemize}
}
\end{corollary}
\begin{proof}
In the case of (1), we know from Corollary \ref{mainthm02}
that $\sigma_{\rm p}(H) \not = \emptyset$.
By Proposition \ref{prop_loc},
we need only take the initial state $\Psi_0 \in \mathcal{H}_{sc}(H)^\perp$ 
overlapping with $\mathcal{H}_{p}(H)$.
In the case of (2), $\sigma_{\rm sc}(H) = \emptyset$
and any initial state $\Psi_0$ overlaps with $\mathcal{H}_{p}(H)$.
\end{proof}

\section{Generator of an evolution}
\label{sectiongenerator}
In this section, we prove Theorem \ref{mainthm01} and Corollary \ref{mainthm02}.
We begin with the precise definition of notations.
\subsection{Definition and properties of $d_\pm$}
\label{sec.3.1}
Let $\vartheta:[-1,1] \to [0,\pi]$ be a function defined by
\[ \vartheta(\lambda) = \arccos \lambda, \quad \lambda \in [-1,1]. \]
Because $\sigma(T)\subseteq [-1,1]$, 
\[ \cos \vartheta (T) = T, \quad \sin \vartheta(T) = \sqrt{1-T^2},
	\quad e^{\pm i \vartheta(T)} = T \pm i \sqrt{1-T^2}. \]
Note that
$\ker (T^2-1) = \ker \sqrt{1-T^2}$ and 
$\ker(T^2-1)^\perp = \overline{{\rm Ran}\sqrt{1-T^2}}$.
We first define operators
$d_\pm^\dagger:\mathrm{Ran}(T^2-1) \to \mathcal{D}_1$
as follows:
for $f \in {\rm Ran}(T^2-1)$,
\begin{align*} 
d_+^\dagger f
	= (d_A^* - d_B^* e^{i\vartheta(T)}) \frac{1}{\sqrt{2(1-T^2)}}f, \quad
d_-^\dagger f
	= (d_A^* e^{i\vartheta(T)} - d_B^* ) \frac{1}{\sqrt{2(1-T^2)}}f.
\end{align*} 
Because
$\frac{1}{\sqrt{2(1-T^2)}}f \in {\rm Ran}\sqrt{1-T^2}$
for all $f \in {\rm Ran}(1-T^2)$,
we know that  $d_\pm^\dagger f \in \mathcal{D}_1$.
\begin{lemma}
\label{0401lemma01}
{\rm
$d_\pm^\dagger$ are isometries from ${\rm Ran}(T^2-1)$ to $\mathcal{D}_1$.
}
\end{lemma}
\begin{proof}
Because by direct calculation, 
\[ (d_A - e^{-i\vartheta(T)}d_B ) (d_A^* - d_B^* e^{i\vartheta(T)}) = 
		 (2 - 2 T \cos\vartheta(T)) = 2(1-T^2) \]
it follows that for all $f \in {\rm Ran}(T^2-1)$, 
\begin{align*}
\|d_+^\dagger f\|^2
	& = \left\langle  \frac{1}{\sqrt{2(1-T^2)}}f, 
		 (d_A - e^{-i\vartheta(T)}d_B ) (d_A^* - d_B^* e^{i\vartheta(T)}) 
		 	\frac{1}{\sqrt{2(1-T^2)}}f \right\rangle \\
	& = \|f\|^2.
\end{align*}
This implies that $d_+^\dagger$ is an isometry on ${\rm Ran}(T^2-1)$.
Noting that $ (e^{-i\vartheta(T)}d_A - d_B ) (d_A^* e^{i\vartheta(T) - d_B^*})=2(1-T^2)$,
we also know that $d_-^\dagger$ is an  isometry on ${\rm Ran}(T^2-1)$.
\end{proof}
From Lemma \ref{0401lemma01},
$d_\pm^\dagger$ have unique extensions,
whose domains are
$\overline{{\rm Ran}(T^2-1)}$ $=\ker(T^2-1)^\perp$.
We denote the extension by the same symbol, 
{\it i.e.}, $d_\pm^{\dagger}: \ker (T^2-1)^\perp \to \mathcal{D}_1$ is given by
\[ d_\pm^\dagger f = \lim_{n \to \infty} d_\pm^\dagger f_n,
	\quad f \in \ker(T^2-1)^\perp, \]
where $\{f_n\} \subset {\rm Ran}(T^2-1)$ is an arbitrary sequence
satisfying $\lim_n f_n = f$.
Thus, we have the following:
\begin{proposition}
{\rm
$d_\pm^\dagger$ are isometries from $\ker(T^2-1)^\perp$ to $\mathcal{D}_1$.
}
\end{proposition}

We use $\mathcal{D}_1^\pm$ to denote the range of $d_\pm^\dagger$:
\[ \mathcal{D}_1^\pm = d_\pm^\dagger \ker(T^2-1)^\perp. \]
\begin{lemma}
\label{lemmaDpm}
{\rm
$\mathcal{D}_1^{\pm}$ are closed subspaces of ${\mathcal{D}_1}$
and 
\[ {\mathcal{D}_1} = \mathcal{D}_1^+ \oplus \mathcal{D}_1^-. \]
}
\end{lemma}
\begin{proof}
{\rm
Because $d_\pm^\dagger$ is an isometry,
it is clear that $\mathcal{D}_1^\pm$ is a closed subspace of $\mathcal{D}$.
We first show that $\mathcal{D}_1^\pm$ are orthogonal to each other.
Let
$\psi_\pm \in \mathcal{D}_1^\pm$
and write it as
$\psi_\pm = \lim_{n \to \infty} d_\pm^\dagger f_n^{\pm}$ ($f_n^\pm \in {\rm Ran}(T^2-1)$).
It follows that
\begin{align*}
& \langle\psi_+, \psi_- \rangle
	= \lim_{n\to \infty} \left\langle 
			d_+^\dagger f_n^+, d_-^\dagger f^-  
				\right\rangle \\
	& \quad =  \lim_{n\to \infty} \left\langle  
		 \frac{1}{\sqrt{2(1-T^2)}}f_n^+,
		(d_A - e^{-i\vartheta(T)} d_B ) (d_A^* e^{i\vartheta(T)} - d_B^* ) \frac{1}{\sqrt{2(1-T^2)}}f_n^- 
			\right\rangle \\
			& \quad = 0,
\end{align*}
where in the last equality, we have used the fact that
\begin{equation} 
\label{eqdperp}
(d_A - e^{-i\vartheta(T)} d_B ) (d_A^* e^{i\vartheta(T)} - d_B^* )
	= 2 \cos \vartheta(T)-2T = 0.
\end{equation}

It remains to be shown that 
${\mathcal{D}_1} =  \mathcal{D}_1^+ \oplus \mathcal{D}_1^- $.
It suffices to show that
$d_A^* \ker(T^2-1)^\perp + d_B^* \ker(T^2-1)
  	\subset \mathcal{D}_1^+ \oplus \mathcal{D}_1^- $.
To this end, take a $\psi \in d_A^* \ker(T^2-1)^\perp + d_B^* \ker(T^2-1)$.
From \cite{SS15}, 
there exist unique vectors $ f, \ g \in \ker(T^2-1)^\perp$
such that
\[ \psi = d_A^* f + d_B^* g. \]
We now take vectors $f_n, \ g_n \in {\rm Ran}\sqrt{1-T^2}$
satisfying $f=\lim_{n \to \infty} f_n$ and $g=\lim_{n \to \infty} g_n$
and set
\begin{align*}
F_n = -\frac{1}{\sqrt{2} i}(e^{-i\vartheta(T)} f_n + g_n), \quad
	G_n = \frac{1}{\sqrt{2} i}(f_n +e^{-i\vartheta(T)} g_n).
\end{align*}
Then, $F_n, \ G_n \in {\rm Ran}\sqrt{1-T^2}$
and
\begin{align*}
f_n = \frac{1}{\sqrt{2(1-T^2)}} (F_n + e^{i\vartheta(T)} G_n), \quad
	g_n = - \frac{1}{\sqrt{2(1-T^2)}}  (e^{i\vartheta(T)}F_n +  G_n). 
\end{align*}
By direct calculation, 
\begin{align*}
d_+^\dagger F_n + d_-^\dagger G_n
	& = d_A^* f_n + d_B^* g_n.
\end{align*}
Since the limits $F:=\lim_{n \to \infty} F_n$ and $G:=\lim_{n \to \infty} G_n$ exist
and $F, \ G \in \ker(T^2-1)^\perp$,
\begin{align*}
\psi & = \lim_{n \to \infty} (d_A^* f_n + d_B^* g_n) 
		= \lim_{n \to \infty} (d_+^\dagger F_n + d_-^\dagger G_n) \\
	& = d_+^\dagger F + d_-^\dagger G \in \mathcal{D}_1^+ \oplus \mathcal{D}_1^-.
\end{align*}
This completes the proof.
}
\end{proof}

Let $d_{\pm,1}$ be the adjoint of $d_\pm^\dagger:\ker(T^2-1)^\perp \to \mathcal{D}_1$.
Then,
\[ d_{\pm,1} = (d_\pm^\dagger)^*, \quad d_{\pm,1}^* = d_\pm^\dagger. \]
\begin{proposition}
\label{prop4231708}
{\rm
On the entire ${\mathcal{D}_1}$,
\begin{align*}
d_{+,1} = \frac{1}{\sqrt{2(1-T^2)} }(d_A -e^{-i\vartheta(T)} d_B), \quad
	d_{-,1} = \frac{1}{\sqrt{2(1-T^2)} }(e^{-i\vartheta(T)} d_A - d_B).
\end{align*}
Moreover, 
\begin{itemize}
\item[(i)] $d_{\pm,1} d_{\pm,1}^* = I_{\ker(T^2-1)^\perp}$, \ $d_{\pm,1} d_{\mp,1}^* = 0$.
\item[(ii)] $\tilde\Pi_{\mathcal{D}_1^\pm} := d_{\pm,1}^*d_{\pm,1}$ is 
	the projection from $\mathcal{D}_1$ onto $\mathcal{D}_1^\pm$.
\end{itemize}
}
\end{proposition}
To prove this proposition,
we use the following lemma:
\begin{lemma}
\label{lemmadpmdmp*}
{\rm
\begin{itemize}
\item[(i)] $(d_A - e^{-i\vartheta(T)} d_B ) (d_A^* e^{i\vartheta(T)} - d_B^* ) =0$.
\item[(ii)] $( e^{-i\vartheta(T)} d_A -d_B ) (d_A^* - d_B^* e^{i\vartheta(T)} )= 0$.
\item[(iii)] $( d_A -e^{-i\vartheta(T)}d_B ) (d_A^* - d_B^* e^{i\vartheta(T)} )= 2(1-T^2)$.
\item[(iv)]  $(e^{-i\vartheta(T)} d_A - d_B ) (d_A^* e^{i\vartheta(T)} - d_B^* ) = 2(1-T^2)$.
\end{itemize}
}
\end{lemma}
\begin{proof}
{\rm
(i) is proved in \eqref{eqdperp}.
(ii) is obtained from (i) by taking the adjoint.
(iii) is also obtained from the adjoint of (iv).
(iv) is proved by direct calculation:
\begin{align*}
(e^{-i\vartheta(T)} d_A - d_B ) (d_A^* e^{i\vartheta(T)} - d_B^* )
= 2 - 2 T \cos \vartheta(T) = 2(1-T^2).
\end{align*}
}
\end{proof}
\begin{proof}[Proof of Propositon \ref{prop4231708}]
{\rm
For all $F \in \ker (T^2 -1)^\perp$,
there exists a sequence $\{F_n\} \subset {\rm Ran}\sqrt{1-T^2}$
such that $F=\lim_{n \to \infty}F_n$.
From (iii) and (iv) of Lemma \ref{lemmadpmdmp*},
\begin{align}
( d_A -e^{-i\vartheta(T)}d_B ) d_+^\dagger F
	& =\lim_{n \to \infty} \sqrt{2(1-T^2)} F_n \notag \\
\label{122301} 
	& = \sqrt{2(1-T^2)} F \in {\rm Ran} \sqrt{1-T^2}, 
	\\
(e^{-i\vartheta(T)} d_A - d_B ) d_-^\dagger  F
	& =\lim_{n \to \infty} \sqrt{2(1-T^2)} F_n \notag \\
\label{122302}
	& = \sqrt{2(1-T^2)} F \in {\rm Ran} \sqrt{1-T^2}.
\end{align}
In addition, from (i) and (ii) of Lemma \ref{lemmadpmdmp*},
\begin{align}
\label{122303}
( d_A -e^{-i\vartheta(T)}d_B ) d_-^\dagger  F
	& =0, \\
\label{122304}
(e^{-i\vartheta(T)} d_A - d_B ) d_+^\dagger  F
	& =0.
\end{align}
By \eqref{122301}, \eqref{122302}, \eqref{122303} and \eqref{122304},
we know that the operators 
$o_+ :=\frac{1}{\sqrt{2(1-T^2)} }(d_A -e^{-i\vartheta(T)} d_B)$
and $o_-:=\frac{1}{\sqrt{2(1-T^2)} }(e^{-i\vartheta(T)} d_A - d_B)$
can be defined on the entire $\mathcal{D}_1$.
To prove that $d_{\pm,1} = o_\pm$,
it suffices to show that the adjoint of $o_\pm$ are $d_\pm^\dagger$.
For all $\psi \in \mathcal{D}_1$ and
$f \in \ker(T^2-1)^\perp$,
\begin{align*}
\left\langle f, o_+ \psi \right\rangle 
& = \lim_{n \to \infty} 
	\left\langle \frac{1}{\sqrt{2(1-T^2)}} f_n,  (d_A -  e^{-it \vartheta(T)} d_B) \psi 
		 \right\rangle \\
& = \lim_{n \to \infty} 
	\left\langle (d_A^* - d_B^*  e^{it \vartheta(T)} ) \frac{1}{\sqrt{2(1-T^2)}} f_n,   \psi 
		 \right\rangle 
	= \langle d_+^\dagger f, \psi \rangle, 
\end{align*}
where $\{f_n\} \subset {\rm Ran}(T^2-1)$
is a sequence such that $f =\lim_{n \to \infty}f_n$.
This means that  $d_+^\dagger$
is the adjoint of $o_+$. 
Hence, $d_{+,1}=o_+$.
The same proof works for $d_{-,1} = o_-$.
The former statement of the proposition is proved.

(i) is proved from Lemma \ref{lemmadpmdmp*}.
We prove (ii). To this end,  we take $\psi_\pm \in \mathcal{D}_1^\pm$
and write it as
$\psi_\pm = d_\pm^\dagger  F$ ($F \in \ker(T^2-1)^\perp$).
Combining (i) with $d_{\pm,1}^* = d_\pm^\dagger$ yields the result that
\[ \tilde\Pi_{\mathcal{D}_1^\pm} \psi_\pm = (d_{\pm,1}^* d_{\pm,1}) (d_\pm^\dagger  F) 
	= d_\pm^* F = \psi_\pm. \] 
Hence, ${\rm Ran}\Pi_{\mathcal{D}_1^\pm} = \mathcal{D}_1^\pm$.
It remains to be proved that $\tilde\Pi_{\mathcal{D}_1^\pm}$ is a projection.
It is clear, by definition, that $\tilde\Pi_{\mathcal{D}_1^\pm}$ is self-adjoint.
By (i), 
$\tilde\Pi_{\mathcal{D}_1^\pm}^2 = d_{\pm,1}^*(d_{\pm,1} d_{\pm,1}^*) d_{\pm,1}
	= \tilde\Pi_{\mathcal{D}_1^\pm}$, 
and we obtain the desired result.
}
\end{proof}

In what follows, 
we extend the domain $\mathcal{D}_1$ of $d_{\pm,1}$ 
to the entire space $\mathcal{H}$. 
We will denote the extension of $d_{\pm,1}$ by $d_\pm$.
\begin{lemma}
\label{lem42910:42}
{\rm
On $\mathcal{D}_0 \oplus \mathcal{D}^\perp$,
\begin{itemize}
\item[(i)] $d_A - e^{-i\vartheta(T)} d_B = 0$;
\item[(ii)] $e^{-i\vartheta(T)}d_A -  d_B = 0$.
\end{itemize}
}
\end{lemma}
\begin{proof}
Because by (iii) of Proposition \ref{prop42910:04},  (i) and (ii) hold on $\mathcal{D}^\perp$,
we need to only establish them on $\mathcal{D}_0$.
Let $\psi_0 \in \mathcal{D}_0$
and write it as
$\psi_0 = d_A^* f_0$ ($f_0 \in \ker(T^2-1)$).
Then,
\begin{align*} 
(d_A - e^{-i\vartheta(T)} d_B)\psi_0
	& = (1 -e^{-i\vartheta(T)} T) f_0 \\
	& = i\sqrt{1-T^2} e^{-i\vartheta(T)} f_0 = 0.
\end{align*}
Similarly, 
\begin{align*} 
(e^{-i\vartheta(T)}d_A -  d_B)\psi_0
	= -i\sqrt{1-T^2} f_0 = 0.
\end{align*}
\end{proof}
By Lemma \ref{lem42910:42}, 
operators $d_\pm:\mathcal{H} \to \mathcal{K}$ can be defined by
\begin{align*}
d_{+} = \frac{1}{\sqrt{2(1-T^2)} }(d_A -e^{-i\vartheta(T)} d_B), \quad
	d_{-} = \frac{1}{\sqrt{2(1-T^2)} }(e^{-i\vartheta(T)} d_A - d_B)
\end{align*}
and 
\begin{equation}
\label{eq42912:00}
d_\pm = d_{\pm,1} \Pi_{\mathcal{D}_1^\pm}, \quad
	d_\pm^* = d_{\pm,1}^* \Pi_{\ker(T^2-1)^\perp},
\end{equation}
where $\Pi_{\mathcal{D}_1^\pm}$ and $\Pi_{\ker(T^2-1)^\perp}$ are 
the projections onto $\mathcal{D}_1^\pm$ and $\ker(T^2-1)^\perp$,
respectively. 
From \eqref{eq42912:00} and Proposition \ref{prop4231708}, 
we have the following:
\begin{proposition}
\label{prop42914:25}
{\rm
Let $d_\pm$ be defined as above.
\begin{itemize}
\item[(i)] $\ker(d_\pm) 
	= \mathcal{D}_1^\mp \oplus 
		\mathcal{D}_0 \oplus \mathcal{D}^\perp$
	and ${\rm Ran}(d_\pm) = \ker(T^2-1)^\perp$;
\item[(ii)] $\mathcal{D}_1^\pm = d_\pm^* \ker(T^2-1)^\perp$;
\item[(iii)] $d_\pm d_\pm^* = \Pi_{\ker(T^2-1)^\perp}$,  $d_\pm d_\mp^*=0$;
\item[(iv)] $d_\pm^*d_\pm = \Pi_{\mathcal{D}_1^\pm}$.
\end{itemize}
}
\end{proposition}

\subsection{Generator of $U$}
\label{sec.3.2}
By Proposition \ref{prop0409}, the evolution $U=S(2d_A^*d_A-1)$
associated with $d_A$ and $S$ is decomposed as
\begin{equation}
\label{eqU429} 
U = U_{\mathcal{D}_1} \oplus I_{\ker (U-1)} \oplus (-I_{\ker(U+1)}),
\end{equation}
where $\mathcal{D}_1 = \ker(U^2 -1)^\perp$ 
and $\ker(U \mp 1) = \mathcal{D}_0^\pm \oplus \mathcal{D}_\pm^\perp$. 
We first prove the following representation of $U_{\mathcal{D}_1}$:
\begin{theorem}
\label{Mainthm04}
{\rm
Let $U$ be as above. 
$U$ leaves $\mathcal{D}_1^\pm$ invariant, and 
$U_{{\mathcal{D}_1}}$ is decomposed as
\[ U_{\mathcal{D}_1}
	 	= e^{i \vartheta(d_+^* Td_+)} \oplus e^{-i \vartheta(d_-^* Td_-)} 
	 		\quad \mbox{on $\mathcal{D}_1 
	 			= \mathcal{D}_1^+ \oplus \mathcal{D}_1^-$}. \]
}
\end{theorem}
\begin{proof}
{\rm
Let  $\psi \in {\mathcal{D}_1}$.
Because by Proposition \ref{prop42914:25}, 
$d_\pm \psi \in \ker(T^2-1)^\perp$,
we know that
there exists  a sequence $\{F_n^\pm\} \subset {\rm Ran}(T^2-1)$ 
such that $d_\pm \psi = \lim_n F_n^\pm$.
Hence,
\begin{align*}
U(\Pi_{\mathcal{D}_1^+}\psi) & = \lim_n Ud_+^* F_n^+ \\
	& = \lim_n 
		 U (d_A^* - d_B^* e^{i\vartheta(T)}) \frac{1}{\sqrt{2(1-T^2)}} F_n^+ \\
	& =  \lim_n 
		 \left(d_A^*
		 	+ d_B^* (e^{-i\vartheta(T)}-2T)\right) \frac{1}{\sqrt{2(1-T^2)}} 
		 		 e^{i\vartheta(T)}F_n^+,
\end{align*}
where we have used the facts that
$U d_A^* = d_B^*$ and $Ud_B^* = 2d_B^* T-d_A^*$.
Because $e^{-i\vartheta(T)}-2T = - e^{i\vartheta(T)}$,
it follows that
\begin{equation}
\label{eq429(1)} 
U(\Pi_{\mathcal{D}_1^+}\psi) = \lim_n d_+^* e^{i\vartheta(T)} F_n^+ 
	= d_+^* e^{i\vartheta(T)}d_+ \psi \in \mathcal{D}_1^+, 
\end{equation}
which proves that $U$ leaves $\mathcal{D}_1^+$ invariant.
Similarly, using $e^{i\vartheta(T)}-2T = - e^{-i\vartheta(T)}$
yields the result that
\begin{align}
U(\Pi_{\mathcal{D}_1^-}\psi) & = \lim_n Ud_-* F_n^- \notag \\
	& = \lim_n 
		 U (d_A^*e^{i\vartheta(T)} - d_B^* ) \frac{1}{\sqrt{2(1-T^2)}} F_n^- \notag \\
	& =  \lim_n 
		 \left(d_A^* e^{i\vartheta(T)} 
		 	+ d_B^*(e^{i\vartheta(T)}-2T) e^{i\vartheta(T)}\right) 
		 		\frac{1}{\sqrt{2(1-T^2)}} 
		 		 e^{-i\vartheta(T)}F_n^+ \notag \\
	& =  d_-^* e^{-i\vartheta(T)}d_- \psi \in \mathcal{D}_1^-.
		\label{eq429(2)} 
\end{align}
Hence, the former half of the theorem follows.
By \eqref{eq429(1)} and \eqref{eq429(2)},
it follows that for all $\psi \in \mathcal{D}_1$,
\begin{align*}
U \psi 
	& = U (\Pi_{\mathcal{D}_1^+}\psi) + U(\Pi_{\mathcal{D}_1^-}\psi) \\
	&  = d_+^* e^{i\vartheta(T)}d_+ \psi 
		+ d_-^* e^{-i\vartheta(T)}d_- \psi.
\end{align*}
Because by Proposition \ref{prop42914:25}, 
$d_\pm:\mathcal{D}_1^\pm\to \ker(T^2-1)^\perp$ is unitary,
\[ d_\pm^* e^{\pm i\vartheta(T)}d_\pm =  e^{\pm i\vartheta(d_\pm^*Td_\pm)}, \]  
which completes the proof.
}
\end{proof}
\begin{proof}[Proof of Theorem \ref{mainthm01}]
Let $H$ be defined by \eqref{eqHform}. 
\begin{align*}
e^{iH} 
& = e^{i \vartheta(d_+^*Td_+)} \oplus e^{i (2\pi - \vartheta(d_-^*Td_-))}
			\oplus e^{0} \oplus e^{i\pi} \\
& = e^{i \vartheta(d_+^*Td_+)} \oplus e^{- i \vartheta(d_-^*Td_-)}
			\oplus 1 \oplus (-1)
\end{align*}
on 
$\mathcal{H}
	= \mathcal{D}_1^+ \oplus \mathcal{D}_1^- 
		\oplus \ker (U-1) \oplus \ker(U+1)$. 
By \eqref{eqU429}, $e^{iH} = U$.
Because {$E_H([0,2\pi)) = I$},
we obtain the desired result.
\end{proof}
\begin{proof}[Proof of Corollary \ref{mainthm01.5}]
Because $\psi_0 \in {\rm Ran}(d_+^*d_+) = \mathcal{D}_1^+$,
\[ f_n = d_+U^n \psi_0 = e^{in\vartheta(T)} d_+ \psi_0. \]
Hence,
\[ \frac{1}{2}(f_n + f_{n-1})
	= \frac{e^{i\vartheta(T)} + e^{-i\vartheta(T)}}{2} e^{in\vartheta(T)} d_+ \psi_0 
	= T f_n. \]
\end{proof}
\begin{proof}[Proof of Corollary \ref{mainthm02}]
Let $T_1 = T\Pi_{\ker(T^2-1)^\perp}$. 
Because $H$ has of the form \eqref{eqHform},
it follows that
\[ \sigma_{\rm p}(H)
	= \{ \vartheta_+(\lambda) \mid \lambda \in \sigma_{\rm p}(T_1) \} 
		\cup 
	\{ \vartheta_-(\lambda) \mid \lambda \in \sigma_{\rm p}(T_1) \}
		\cup  \{0, \pi\}. \]
Here, we set $\vartheta_+=\vartheta$ and $\vartheta_-=2\pi-\vartheta$.
It is clear that $\ker(H) = \ker(U-1)$ and $\ker(H- \pi) = \ker(U+1)$.
Because $d_\pm:\mathcal{D}_1^\pm \to \ker(T^2-1)^\perp$ are unitary,
\[ \ker(H-\vartheta_\pm(\lambda))
	= d_\pm^* \ker(T- \lambda). \]
Hence,
\begin{align*} 
\mathcal{H}_{\rm p}(H)
	& = \left[\bigoplus_{\lambda \in \sigma_{\rm p} (T_1)} 
		d_+^* \ker(T- \lambda) \right]
	\oplus \left[\bigoplus_{\lambda \in \sigma_{\rm p} (T_1)} 
		d_-^* \ker(T- \lambda) \right] 
	\oplus \ker (U^2-1) \\
	& = d_+^* \mathcal{H}_{\rm p}(T_1) \oplus d_-^* \mathcal{H}_{\rm p}(T_1)
		\oplus \ker (U^2-1).
\end{align*}
Because $\mathcal{H}_{\rm p}(T_1) = \mathcal{H}_{\rm p}^T$,
we obtain the former statement of the corollary.
The latter follows from  
$\mathcal{H}_{\rm p}(T)^\perp
	= \mathcal{H}_{\rm c}(T) 
	=\mathcal{H}_{\rm ac}(T) \oplus \mathcal{H}_{\rm sc}(T) $
and the unitarity of $d_\pm$. 
\end{proof}

\section{Conclusion}
In this paper,
we gave the explicit formula of the generator $H$
of the abstract Szegedy evolution operator $U = S(2d_A^*d_A-1)$
in terms of the discriminant operator $T = d_A S d_A^*$.
Using this formula,
we characterized the spectral properties of $H$.
By the discrete analog of the RAGE theorem,
we also characterized the asymptotic properties of a quantum walker
in terms of the generator $H$.
In the case of the abstract Szegedy walk,
we obtained the criteria for localization in terms of $T$ and the subspace
$\mathcal{D}^\perp$.
In particular, for the Grover walk on a symmetric graph $G$,
this implies that
localization occurs 
for some initial state $\Psi_0$
only when the transition operator $P_G$ has 
an eigenvalue or $\mathcal{D}_\pm^\perp \not=\emptyset$.
In our future work,
we will apply the theory developed in this paper
to an inhomogeneous QW on $\mathbb{Z}$ such as \cite{K,ShK}.

We also gave the axiom of the abstract discrete-time QWs,
which includes many QWs. 
Given a unitary operator $U$ on a Hilbert space $\mathcal{H}$
and a decomposition $\mathcal{H} = \bigoplus_{x \in V}\mathcal{H}_x$,
we can naturally define a directed graph $G_U$ with vertices $V$
and the finding probability of a quantum walker moving on $G_U$. 
In forthcoming papers,  
we will treat the following problems:
\begin{itemize}
\item[(1)] What kind of unitary operator $U$ has a boundary operator
$d_A$ and a shift operator $S$ such that $U=S(2d_A^*d_A-1)$?
\item[(2)] What is the graph $G_U$? 
\end{itemize}


\vspace{2mm}
\noindent
{\bf Acknowledgements} \quad
The authors thank H. Ohno and Y. Matsuzawa 
for their useful comments.
ES and AS also acknowledge financial supports of the Grant-in-Aid for Young Scientists (B) of Japan Society for the Promotion of Science (Grants No. 25800088 and No. 26800054, respectively). 
ES is also supported by the Japan-Korea Basic
Scientific Cooperation Program ``Non-commutative Stochastic Analysis:
New Prospects of Quantum White Noise and Quantum Walk" (2015-2016).

\appendix
\section{Appendix}

\subsection{Proof of Proposition \ref{proplocal}}
\label{sec.a.1}
We present a proof of Proposition \ref{proplocal}.
Let $H$ be the generator of  an evolution 
$(U, \{\mathcal{H}_v \}_{v \in V}) \in \mathscr{F}_{\rm QW}$. 
Throughout this subsection,  
we assume that ${\rm dim}\mathcal{H}_v < \infty$ ($v \in V$).
Let $\mathcal{H}_1$ be the set of vectors $\Psi_0 \in \mathcal{H}$ satisfying
\[  \lim_{N \to \infty} \frac{1}{N}\sum_{n=0}^{N-1} \nu_n^{\Psi_0}(R) = 0 \]
for any finite subset $R$ of $V$,
and $\mathcal{H}_2$ the set of vectors  $\Psi_0 \in \mathcal{H}$ satisfying
\[ \lim_{m \to \infty} \sup_{n} \nu_n^{\Psi_0}(R_m^{\rm c}) = 0 \]
for any sequence $\{R_m\}$ of finite subsets of $V$  such that 
$R_m \subset R_{m+1}$  and $V=\cup_m R_m$.
Because
$\nu_n^{\alpha \Psi_0 + \beta \Phi_0}(R) 
	\leq 2 \left( |\alpha|^2 \nu_n^{\Psi_0}(R) 
		+ |\beta|^2 \nu_n^{\Phi_0}(R) \right)$, 
we know that $\mathcal{H}_1$ and $\mathcal{H}_2$ are subspaces of $\mathcal{H}$.
Let $P_R = \sum_{x \in R} P_x$ ($R \subset V$).
Then,
\[ \nu_n^{\Psi_0}(R) =  \left\|P_R e^{in H} \Psi_0 \right\|^2. \]
\begin{lemma}
\label{lem42922:34}
{\rm
$\mathcal{H}_1 \perp \mathcal{H}_2$.
}
\end{lemma}
\begin{proof}
Let $\Psi_0 \in \mathcal{H}_1$ and $\Phi_0 \in \mathcal{H}_2$.
Then, for all $R \subset V$,
\begin{align*}
|\langle \Psi_0, \Phi_0 \rangle| 
& = \frac{1}{N} \sum_{n=0}^{N-1} |\langle \Psi_n, \Phi_n \rangle| \\
& \leq \frac{1}{N} \sum_{n=0}^{N-1} 
		 |\langle P_R \Psi_n, P_R \Phi_n \rangle|
	+ \frac{1}{N} \sum_{n=0}^{N-1}
		 |\langle P_{R^{\rm c}} \Psi_n, P_{R^{\rm c}}  \Phi_n \rangle| \\
& \leq \|\Phi_0\| \left( \frac{1}{N} \sum_{n=0}^{N-1}   \|P_R \Psi_n\| \right)
		+  \|\Psi_0\| \left( \frac{1}{N} \sum_{n=0}^{N-1}  \|P_{R^{\rm c}}  \Phi_n\| \right).
\end{align*}
We first estimate the first term.
By the Cauchy-Schwarz inequality, 
\begin{align*}
\frac{1}{N} \sum_{n=0}^{N-1}   \|P_R \Psi_n\|
	\leq \left(\frac{1}{N} \sum_{n=0}^{N-1}   \|P_R \Psi_n\|^2 \right)^{1/2}
	= \bar \nu_N^{\Psi_0}(R)^{1/2}.
\end{align*}
The second term is estimated as follows:
\[ \frac{1}{N} \sum_{n=0}^{N-1}   \|P_{R^{\rm c}} \Phi_n\|
	\leq \sup_{n \geq 0} \|P_{R^{\rm c}} \Phi_n\| 
		= \sup_{n \geq 0} \nu_n^{\Phi_0}(R^{\rm c})^{1/2}. \]
Combining these inequalities yields the result that
\begin{equation}
\label{eq42921:57} 
|\langle \Psi_0, \Phi_0 \rangle |
	\leq \|\Phi_0\|  \bar\nu_N^{\Psi_0}(R)^{1/2}
		+  \|\Psi_0\| \sup_{n \geq 0} \nu_n^{\Phi_0}(R^{\rm c})^{1/2}. 
\end{equation}
Let $\epsilon >0$ and $\{R_m\}_{m\geq 1}$ be a family of finite subsets of $V$
such that $R_m \subset R_{m+1}$ and $V = \cup_{m \geq 1} R_m$.
Because $\Phi_0 \in \mathcal{H}_2$, 
there exists an $m_0 \in \mathbb{N}$ such that 
$\nu_n^{\Phi_0}(R_m^{\rm c}) < \epsilon^2/\|\Psi_0\|^2$ ($m \geq m_0$).
Because $\Psi_0 \in \mathcal{H}_1$,
it follows from \eqref{eq42921:57} that
\begin{align*}
\lim_{N \to \infty} |\langle \Psi_0, \Phi_0 \rangle |
	\leq \epsilon,
\end{align*}
which completes the proof.
\end{proof}

\begin{lemma}
\label{lem42922:35}
{\rm
\begin{itemize}
\item[(i)] $\mathcal{H}_{\rm c}(H) \subset \mathcal{H}_1$;
\item[(ii)] $\mathcal{H}_{\rm p}(H) \subset \mathcal{H}_2$.
\end{itemize}
}
\end{lemma}
\begin{proof}
Let $\Psi_0 \in \mathcal{H}_{\rm c}(H)$.
For any finite set $R$,
\begin{equation}
\label{eq050215:07}
\bar\nu_N^{\Psi_0}(R) 
	= \sum_{x \in R} \sum_{j =1}^{{\rm dim}\mathcal{H}_x}  \bar\nu_N (\phi_{x,j}), 
\end{equation}
where $\{\phi_{x,j}\}$ is a complete orthonormal system of $\mathcal{H}_x$
and 
$\bar\nu_N (\phi)$
	$:= \frac{1}{N} \sum_{n=0}^{N-1}$ $|\langle \phi, e^{inH}\Psi_0 \rangle|^2$.
Because, by assumption, 
the sum in \eqref{eq050215:07} runs over a finite set,
it suffices to show that $\lim_{N \to \infty} \bar\nu_N (\phi) = 0$.
Let $\omega(x) = e^{in x}$ and 
{
$g_N(\omega) = \frac{1}{N}\sum_{n=0}^{N-1} \omega^n$.
Then, $g_N(\omega) = \frac{1-\omega^N}{N(1-\omega)}$ if $\omega \not=1$
and $g_N(1) = 1$.}
By the Fubini theorem,
{
\[ \bar\nu_N(\phi) 
	= \int_0^{2\pi} \int_0^{2\pi} g_N(\omega(\lambda -\mu))
		d\langle P_{\rm c}(H)\phi, E_H(\lambda) \Psi_0 \rangle
			d\langle \Psi_0, E_H(\mu) P_{\rm c}(H) \phi \rangle, \] 
where $P_{\rm c}(H)$ is the projection onto $\mathcal{H}_{\rm c}(H)$.
By the polarization identity,
there exists $\{\psi_j\}_{j=1,2,3,4} \subset \mathcal{H}_{\rm c}(H)$ such that
\[ \bar\nu_N(\phi) 
	\leq {\rm const.} \sum_{j,k=1,2,3,4} \int_0^{2\pi} \int_0^{2\pi} 
		|g_N(\omega(\lambda-\mu))| 
			d\|E_H(\lambda) \psi_j\|^2 d\|E_H(\mu) \psi_k\|^2. \]}
Because $F_j := \|E_H(\cdot) \psi_j\|^2$ is continuous, 
\begin{align*} 
\int\int_{\{(\lambda, \mu) \mid \lambda=\mu\}}
		dF_j(\lambda) dF_k(\mu)
	& \leq \int_0^{2\pi}  dF_k(\mu) 
		\int_{\mu  -\epsilon}^{\mu + \epsilon}
		dF_j(\lambda) \\
	& =  \int_0^{2\pi} dF_k(\mu)  (F_j(\mu + \epsilon) - F_j(\mu - \epsilon)) 
		\to 0,
\end{align*}		
as $\epsilon \to 0$. 
Because {$\sup_{|\omega|=1} |g_N(\omega)| \leq 1$}	
and $\lim_{N \to \infty} g_N(\omega(\lambda-\mu)) = 0$ ($\lambda \not=\mu$),
we obtain $\lim_{N\to 0} \bar\nu_N(\phi) =0$ by the dominated convergence theorem.
This completes the proof of (i).

Let $\Psi_0 \in \mathcal{H}_{\rm p}(H)$.
For any $\epsilon > 0$, there exist eigenvectors 
$\{ \phi_j\}_{j=1}^M$ ($M \in \mathbb{N}$) of $H$ such that 
$\|\Psi_0 - \sum_{j=1}^M \langle \phi_j, \Psi_0 \rangle \phi_j \| < \epsilon$.
Let $\{R_m\}$ be a sequence of finite subsets of $V$
such that $R_m \subset R_{m+1}$ and $\cup_m R_m = V$.
It follows that
\[ \nu_n^{\Psi_0}(R_m^{\rm c})^{1/2}
	\leq \sum_{j=1}^M 
		|\langle \phi_j, \Psi_0 \rangle | \| P_{R_m^{\rm c}} \phi_j \|
			+ \epsilon, \]
which proves $\lim_{m \to \infty} \sup_n \nu_n^{\Psi_0}(R_m^{\rm c}) = 0$.
Hence we have (ii).
\end{proof}

\begin{proof}[Proof of Proposition \ref{proplocal}]
Combining Lemmas \ref{lem42922:34} and \ref{lem42922:35}
yields the result that
\begin{align*} 
\mathcal{H}_2 \subset \mathcal{H}_1^\perp
	\subset \mathcal{H}_{\rm p}(H) \subset  \mathcal{H}_2, \quad
\mathcal{H}_1 \subset \mathcal{H}_2^\perp
	\subset \mathcal{H}_{\rm c}(H) \subset  \mathcal{H}_1,
\end{align*}
which proves the proposition.
\end{proof}

\subsection{Proof of Equation \eqref{eqloc}}
\label{subsecProfeqloc}
In this subsection, we prove the following:
\begin{lemma}
{\rm Let $(U, \{\mathcal{H}_v\}_{v \in V}) \in \mathscr{F}_{QW}$
and $\Psi_0 \in \mathcal{H}$ satisfy 
\[ \lim_{m \to \infty} \sup_n \nu_n^{\Psi_0}(R_m^{\rm c}) = 0 \]
for an increasing sequence $\{R_m \}$ of finite subsets of $V$.
Then, \eqref{eqloc} holds.
In particular, \eqref{eqloc} holds for all $\Psi_0 \in \mathcal{H}_{\rm p}(H)$.
}
\end{lemma}
\begin{proof}
By assumption,
we know that for any $\epsilon > 0$,
there exists $m_0 \in \mathbb{N}$ such that
$\sup_n \nu_n^{\Psi_0}(R_{m_0}^{\rm c}) < \epsilon$.
Hence,
\begin{equation}
\label{eq2317} 
\limsup_{n \to \infty} \nu_n^{\Psi_0} (R_{m_0}) \geq  1-  \epsilon. 
\end{equation}
If $\limsup_n \nu_n^{\Psi_0}(x) = 0$ for any $x \in R_{m_0}$,
then 
\[ \limsup_n \nu_n^{\Psi_0}(R_{m_0})
	=  \sum_{x \in R_{m_0}} \limsup_n \nu_n^{\Psi_0}(x) = 0, \]
which contradicts \eqref{eq2317}.
Therefore, \eqref{eqloc} holds for some $x \in R_{m_0}$.
\end{proof}


\begin{thebibliography}{999}
\bibitem
{Ah} 
Aharonov, L. Davidovich, L., Zagury, N.: 
	Quantum random walks, 
	{Phy. Rev. A} {\bf 48}, 1687--1690 (1993)
\bibitem
{Am01} 
Ambainis, A.,  Bach, E.,  Nayak, A.,  Vishwanath, A.,  Watrous, J.:
	One-dimensional quantum walks,
	{ACM Symp. Theor. Comupt.},
	37--49 (2001)
\bibitem
{Am03} 
Ambainis, A.:
	Quantum walks and their algorithmic applications,
	{\it Int. J. Quantum Inf.} {\bf 1}, 507--518, (2003)
\bibitem
{AKR} 
	Ambainis, A.,  Kempe, J.,  Rivosh, A.: 
	Coins make quantum walks faster, 
	ACM-SIAM Symp. Discrete Algorithm,
	1099--1108 (2005)
\bibitem{AG}
Amrein, W. O., Georgescu, V.: 
	On the characterization of bound states and scattering states in quantum mechanics,
	Helv. Phys. Acta. {\bf 46}, 635--658 (1973)
\bibitem{CFG}
	Childs, A. M., Farhi, E.,  Gutmann, S.:
	An example of the difference between quantum and classical random walks,
	Quantum Inf. Process. {\bf 1}, 35--43 (2002) 
\bibitem{En}
Enss, V.:
	Asymptotic completeness for quantum mechanical potential scattering I. Short rage potentials,
	Commun. Math. Phys. {\bf 61}, 285--291 (1978) 
\bibitem
{FH} 
Feynman, R.P.,  Hibbs, A.R.: 
	Quantum Mechanics and Path Integrals, 
        {McGraw-Hill, Inc., New York}, 34--36 (1965)
\bibitem
{Gud} 
Gudder, S.: 
	Quantum Probability, 
        {Academic Press Inc., Boston} (1988)
\bibitem
{Gr} 
Grover, L.:
	A fast quantum mechanical algorithm for database search,
	ACM Symp. Theor. Comp.,
	212--219 (1996)
\bibitem
{HKSS13} 
Higuchi, Yu.,  Konno, N.,  Sato, I.,  Segawa, E. :
	Quantum graph walks I: mapping to quantum walks,
	{Yokohama Math. J.} {\bf 59}, 33--55 (2013) 
\bibitem
{HKSS14} 
Higuchi, Yu.,  Konno, N.,  Sato, I.,  Segawa, E. :
	Spectral and asymptotic properties of Grover walks on crystal lattices,
	{J. Funct. Anal.} {\bf 267}, 4197--4235 (2014) 
\bibitem
{SS15} 
Higuchi, Yu.,  Segawa, E.,  Suzuki, A.: 
	Spectral mapping theorem of an abstract quantum walk,
	arXiv:1506.06457.
\bibitem
{K} 
Konno, N.:
	One-dimensional discrete-time quantum walks on random environments, 
	{\it Quantum Inf. Process.} {\bf 8}, 387--399 (2009) 
\bibitem
{Mey} 
Meyer, D:
	From quantum cellular automata to quantum lattice gases, 
	{J. Stat. Phys.}, {\bf 85}, 551--574 (1996) 
\bibitem
{MNRS} 
Magniez, F.,  Nayak, A., Roland, J.,  Santha, M.:
	Search via quantum walk,
	ACM Symp. Theor. Comput.,
	575--584 (2007)
\bibitem
{RS3}
Reed, M.,  Simon, B.:
	{Methods of Modern Mathematical Physics, Vol III}, 
	Academic Press, New York (1979)
\bibitem{Ru}
Ruelle, D.:
	A remark on bound states in potential-scattering theory,
	Nuovo Cimento A {\bf 61}, 655--662 (1969)
\bibitem
{Sz} 
Szegedy, M.: 
	Quantum speed-up of Markov chain based algorithms,
	{Ann. IEEE. Symp. Found.},
	32--41 (2004) 
\bibitem
{Se} 
E. Segawa,
	Localization of quantum walks induced by recurrence properties of random walks,
	{J. Comput. Theor. Nanos.} 
	{\bf 10}, 1583--1590 (2013) 
\bibitem
{SKW} 
Shenvi, N.,  Kempe, J.,  Whaley, K.:
Quantum random-walk search algorithm,
Phys. Rev. A {\bf 67}, 052307 (2003)
\bibitem
{Sh} 
Shikano, Y.:
	From discrete-time quantum walk to continuous-time quantum walk in limit distribution, 
	{J. Comput. Theor. Nanos.} {\bf 10}, 1558--1570 (2013)
\bibitem
{ShK} 
Shikano, Y., Katsura,  H.:
	Localization and fractality in inhomogeneous quantum walks with self-duality,
        {Phys. Rev. E} {\bf 82}, 031122 (2010)  
\bibitem{VA15}
Venegas-Andraca, S. E.:
	Quantum walks: a comprehensive review,
	{Quantum Inf. Process.} {\bf 11}, 1015--1106 (2012).
\bibitem
{Wa} 
Watrous, J.:
	Quantum simulations of classical random walks and undirected graph connectivity,
	{J. Comput. Syst. Sci.} {\bf 62}, 376--391 (2001) 
\bibitem{Wi}
Wiener, N.:
	The Fourier integral and certain of its applications,
	Cambridge Univ. Press, London (1935)
\end{thebibliography}
\end{document}